\documentclass{article}
\usepackage{ijcai22}

\usepackage{times}
\usepackage{soul}
\usepackage[hyphens]{url}
\usepackage[hidelinks]{hyperref}
\usepackage[utf8]{inputenc}
\usepackage[small]{caption}
\usepackage{graphicx}
\usepackage{amsmath}
\usepackage{amsthm}
\usepackage{booktabs}
\usepackage{algorithm}
\usepackage{algorithmic}
\title{Contests to Incentivize a Target Group} 

\author{
Edith Elkind\and
Abheek Ghosh\And
Paul Goldberg\\
\affiliations
Department of Computer Science, University of Oxford
\emails
\{edith.elkind,abheek.ghosh,paul.goldberg\}@cs.ox.ac.uk\\
}

\graphicspath{ {./images/} }
\usepackage{amsmath,amsthm,amsfonts,amssymb,bm,bbm,mathtools}
\usepackage{xcolor}
\usepackage[hyphenbreaks]{breakurl}

\newtheorem{theorem}{Theorem}

\newtheorem{example}{Example}
\newtheorem{definition}{Definition}

\DeclareMathOperator{\bb}{\bm{b}}

\DeclareMathOperator{\bt}{\bm{t}}
\DeclareMathOperator{\bv}{\bm{v}}

\DeclareMathOperator{\bx}{\bm{x}}
\DeclareMathOperator{\by}{\bm{y}}

\DeclareMathOperator{\bE}{\mathbb{E}}
\DeclareMathOperator{\bP}{\mathbb{P}}
\DeclareMathOperator{\bR}{\mathbb{R}}

\DeclareMathOperator*{\argmax}{arg\,max}

\begin{document}

\maketitle

\begin{abstract}
    We study how to incentivize agents in a target group to produce a higher output in the context of incomplete information, by means of rank-order allocation contests. We describe a symmetric Bayes--Nash equilibrium for contests that have two types of rank-based prizes: prizes that are accessible only to the agents in the target group; prizes that are accessible to everyone. We also specialize this equilibrium characterization to two important sub-cases: (i) contests that do not discriminate while awarding the prizes, i.e., only have prizes that are accessible to everyone; (ii) contests that have prize quotas for the groups, and each group can compete only for prizes in their share. For these models, we also study the properties of the contest that maximizes the expected total output by the agents in the target group.
\end{abstract}

\section{Introduction}
Contests---situations where multiple agents compete for valuable prizes---are prevalent in many real-life situations, including sports, R\&D races, crowdsourcing, college admissions and job recruitment.\footnote{See the book by Vojnovic~\shortcite{vojnovic2015contest} for an introduction to contest theory and relevant resources.} 

The dominant paradigm in contest theory is to assume that the principal designs the prize structure so as to incentivize the agents to produce a higher total output. In this model, the designer values the contribution from each agent equally, i.e., the marginal output produced by any agent contributes the same marginal value to the designer's objective. In this paper, we break away from this assumption: instead, we assume that the agents belong to two different groups and the designer may want to prioritize one group---the target group---over the other, non-target group.

Contests that give special importance to agents from a particular group are not uncommon in practice. Conferences give best student paper awards, in addition to best paper awards. 
Many competitions and hackathons are organized to elicit engagement from underrepresented groups: for example, recently Microsoft organized a hackathon as an intervention to get women interested in AI.\footnote{\url{https://news.microsoft.com/europe/features/bridging-the-stem-divide-ai-hackathon-helps-young-women-excel-in-computer-science/}} 
Contests are widely used for crowdsourcing, which is also an important source of training data for machine learning algorithms; in this context, eliciting input from disadvantaged groups is particularly important, as it helps the algorithms to learn decision-making rules that reflect opinions and preferences of such groups. Our work provides a better understanding of how to encourage contributors from such groups.

In this paper, we study an incomplete information (Bayesian) model. We assume that each agent is associated with an \textit{ability}, which captures the amount of output they can produce per unit effort. In an incomplete information model for contests, it is generally assumed that the abilities of the agents are selected i.i.d. from a given distribution $F$. In our model, based on the assumption that the agents belong to one of the two groups, the abilities of the agents in the target and the non-target group are sampled from 
two distinct distributions $F$ and $G$, respectively. Each agent belongs to the target group with a fixed probability $\mu$.\footnote{We assume $\mu$ is known based on historical data on the number of participants from the target group compared to the total number of participants. Similarly, the distributions $F$ and $G$ are also known from historical data.}
The agents know the prize allocation scheme, their own ability, and the prior distributions of other agents' abilities. They act strategically to maximize their expected utility, where the utility of an agent is the prize they receive minus the effort they exert, and reach a Bayes--Nash equilibrium. The contest designer knows the prior distributions of the agents' abilities; therefore, can reason about the equilibrium behavior of the agents. She wants to design the prize allocation scheme so as to elicit equilibrium behavior that optimizes her own objective. We assume that the contest designer has a budget of $1$ to use for the prizes.

The contest ranks the agents based on their outputs and awards the prizes based on the ranks. In our general model, we assume that contests have two sequences of prizes, one available to all agents, and the other available to the target group only. A conference sponsoring a best student paper award in addition to a best paper award is an example of such a contest. We also study two specific variants of this model. In the first variant, all prizes in the contest are equally accessible to all agents, and there are no group-specific prizes. One could come across such a design preference in situations where the contest designer does not want to discriminate among the agents, but still wants to incentivize an under-represented agent group. In this case, the contest does not discriminate while awarding the prizes, but the design of the prize structure (i.e., the value of the first prize, second prize, and so on) can incorporate the information about the ability distribution of the agents. This model captures an equal opportunity employer that treats all applicants equally, but has a preference to increase the diversity in the workplace. In the second variant, the contest has group-specific prizes only, with no prizes available to the entire population; examples include a hackathon for women in CS, or a hiring/scholarship contest with a strict group-specific quota.

We assume that the contest designer wants to maximize the expected output of the agents in the target group only. With minor changes, the techniques can be extended to an objective that maximizes the weighted sum of the group outputs and other similar objectives.

\subsection{Our Contribution}
We initiate a theoretical study of contest design to elicit higher participation by agents from under-represented groups and propose tractable models to do so. We analyze the equilibrium behavior of the agents for three cases: \textit{general prizes only}, where there is a sequence of rank-ordered prizes accessible to all agents, whether in the target or the non-target group; \textit{group-specific prizes only}, where there are two different sequences of prizes, each accessible to one of the groups; \textit{both general and target group-specific prizes}, where there are prizes accessible to both groups as well as prizes accessible to the target group only. The techniques used to analyze the equilibrium in these contests can be extended to related prize structures.

Based on these equilibrium characterizations, we study the properties of the contest that maximizes the total expected output of the agents from the target group. For the setting where all prizes are general, we prove that the optimal contest awards a prize of $1/\ell$ to the first $\ell$ agents and $0$ to the remaining agents, where $\ell$ is a function of the relative frequency of target group agents, $\mu$, and the ability distributions of the target group, $F$, and the non-target group, $G$. We give a closed-form formula for $\ell$. We also study the effect of first-order and second-order stochastic dominance of $F$ on $G$, and vice versa, on the optimal contest. For the setting where all prizes are group-specific, we show that it is optimal to award the full prize budget to the top-ranked agent, irrespective of the distributions $F$ and $G$. This result matches similar results on maximizing total output~\cite{moldovanu2001optimal}. We also compare general prizes and group-specific prizes, and show that either of these choices may be preferred depending upon the distributions $F$ and $G$, and also upon the frequency of the target group agents in the population, $\mu$, even if $F=G$.




\subsection{Related Work}
We discuss literature in contest theory directly relevant to our work; we point the readers to the book by Vojnovic~\shortcite{vojnovic2015contest} for a broader survey on contest theory.

Moldovanu and Sela~\shortcite{moldovanu2001optimal} characterize the Bayes--Nash equilibrium for contests with rank-order allocation of prizes assuming incomplete information with i.i.d. types; Chawla and Hartline~\shortcite{chawla2013auctions} prove the uniqueness of this equilibrium. Moldovanu and Sela~\shortcite{moldovanu2001optimal} also show that awarding the entire prize to the top-ranked agent is optimal when the agents have (weak) concave cost functions, but the optimal mechanism may have multiple prizes for convex cost functions. In contrast, we observe that even for linear cost functions, it may be beneficial to have multiple prizes to maximize the expected output from the target group. 

The standard assumption in mechanism design and its sub-area of contest design is that the agents' types are sampled i.i.d. from some distribution. Our paper extends the equilibrium analysis of Moldovanu and Sela~\shortcite{moldovanu2001optimal} to an asymmetric model where the agents are from two groups with different distributions. Amann and Leininger~\shortcite{amann1996asymmetric} study an asymmetric model with two agents with types sampled independently from two different distributions; our analysis uses the idea of the function `$k$' introduced by them, see Theorem~\ref{thm:eqBoth}. Characterizing equilibrium in contests with many agents with types sampled independently from different distributions has remained technically challenging; recently Olszewski and Siegel~\shortcite{olszewski2016large,olszewski2020performance} made progress by assuming very large numbers of agents and prizes, where an individual agent has infinitesimal effect on the equilibrium.

Bodoh-Creed and Hickman~\shortcite{bodoh2018college} model affirmative action in college admissions using contests. Their model has general non-linear utility and cost functions, but, like Olszewski and Siegel~\shortcite{olszewski2016large,olszewski2020performance},
they assume that the numbers of agents and prizes are very large. They give first-order equilibrium conditions and show that two types of affirmative actions---(i) admissions preference schemes, where the outputs of agents from the target group are (artificially) amplified before ranking and prize allocation, and (ii) quotas, where there are separate pools of seats for the different groups (similar to our group-specific prizes only model, Section~\ref{sec:eqTarget})---have the same sets of equilibria with identical actions by the agents and identical outcomes. Our model makes a stronger linearity assumption regarding the utility and the cost functions, which allows us to better characterize the equilibrium and 
study the optimal contest (prize) design problem.

There is a long line of literature on the effect of different types of affirmative actions on contests. This literature generally assumes that there is only one prize, which either gets allocated to the top-ranked agent or gets allocated proportionally (Tullock~\shortcite{tullock1980efficient} contests and their generalizations), and the contest designer introduces different kinds of interventions in the allocation of this prize. One such intervention is \textit{favoritism}, where the objective is to maximize the total output by giving head-starts or handicaps to certain agents (see, e.g., \cite{kirkegaard2012favoritism,fu2020optimal,deng2021optimally}). A head-start 
adds a bonus to an agent's output, while a handicap 
decreases an agent's score by a fixed percentage. See \cite{chowdhury2020heterogeneity} for a survey on contests and affirmative action.


Our work deviates from the widely studied objective of maximizing total output, 
by focusing on the output of a target group. Several papers have studied objectives other than the total output, such as maximum individual output~\cite{chawla2019optimal}, cumulative output from the top $k$ agents~\cite{archak2009optimal,gavious2014revenue}, total output of agents producing output in a given range~\cite{elkind2021contest}. \cite{elkind2021contest} targets agents based on their output level, and therefore, effectively, based on their abilities; our work targets agents based on their group association, with the assumption that exogenous factors may have caused agents in some groups to have acquired less ability.

There are also several papers that do equilibrium analysis and optimal contest design in the complete information setting (e.g.,~\cite{baye1996all,barut1998symmetric,siegel2009all}). 

\section{Model and Preliminaries}\label{sec:prelim}
Let there be $n$ agents. Any given agent belongs to the target (resp., non-target) group with probability $\mu$ (resp., $(1-\mu)$), independently of the other agents. Let $\bv = (v_1, v_2, \ldots, v_n)$ be the ability profile of the agents, where $v_i$ values are independent and identically distributed (i.i.d.) random variables from continuous and differentiable distributions $F$ or $G$ with support $[0,1]$, depending upon whether agent $i$ belongs to the target or the non-target group, respectively. Let $f$ and $g$ be the probability density functions (PDFs) of $F$ and $G$. The $n$ agents simultaneously produce the output profile $\bb = (b_i)_{i \in [n]}$, so that the cost of agent $i$ is $b_i/v_i$.

The contest awards two sequences of prizes $w_1 \ge w_2 \ge \ldots \ge w_n$ and $\omega_1 \ge \omega_2 \ge \ldots \ge \omega_n$, where the prize $w_j$ is given to the $j$-th ranked agent overall (we call these prizes \textit{general} prizes), and the prize $\omega_j$ is given to the $j$-th ranked agent among the agents in the target group (we call these \textit{group-specific} prizes), with ties broken uniformly.\footnote{In our model, it does not matter how we break ties, so w.l.o.g., we can assume uniform tie break. In more detail, for a tie to happen, two players must produce exactly identical output. In our setting, this happens with zero probability, as there are no point masses in the probability distributions $F$ and $G$, and the distributions of the output generated by the agents (that we derive as a result of our equilibrium analysis) also do not have any point masses.} We assume that $\sum_j w_j + \sum_j \omega_j \le 1$, i.e., the contest designer has a unit budget. Given a vector of outputs $\bb$, let the allocation of general prizes be given by $\bx(\bb) = (x_{i,j}(\bb))_{i,j \in [n]}$, where $x_{i,j} = 1$ if agent $i$ is awarded the $j$-th prize and $x_{i,j} = 0$ otherwise. Similarly, let $\by(\bb,\bt) = (y_{i,j}(\bb,\bt))_{i,j \in [n]}$ be the allocation of target group prizes, where $\bt = (t_i)_{i \in [n]}$ is the group label vector; $t_i \in \{T,N\}$, $t_i = T$ if the agent is in the target group, $t_i = N$ if not. We shall suppress the notation for $\bt$ and write $\by(\bb)$ instead of $\by(\bb,\bt)$.

If agent $i$ is in the target group, her utility is given by
\begin{multline}\label{eq:agentUtility1}
u_T(v_i, \bb, \bt) = \sum_{j \in [n]} w_j x_{i,j}(\bb) + \sum_{j \in [n]} \omega_j y_{i,j}(\bb) - b_i/v_i \\
\equiv v_i \left(\sum_{j \in [n]} w_j x_{i,j}(\bb) + \sum_{j \in [n]} \omega_j y_{i,j}(\bb)\right) - b_i,
\end{multline}
where the equivalence is true because scaling the utility function by a constant does not change the strategy of an agent. Similarly, the utility of an agent $i$ in the non-target group is given by
\begin{equation}\label{eq:agentUtility2}
u_N(v_i, \bb, \bt) = v_i \sum_{j \in [n]} w_j x_{i,j}(\bb) - b_i.
\end{equation}

\subsection{Mathematical Preliminaries}
Let $p^H_j(v)$ denote the probability that a value $v \in [0, 1]$ is the $j$-th highest among $n$ i.i.d. samples from a distribution $H$, given by the expression
\begin{equation*}
    p^H_j(v) = \binom{n-1}{j-1} H(v)^{n-j}(1-H(v))^{j-1}.
\end{equation*}

A key role in the Bayes--Nash equilibrium of rank-order allocation contests is played by the order statistics. Let $f^H_{n,j}$ be the PDF of the $j$-th highest order statistic out of $n$ i.i.d. samples from $H$ (with PDF $h$), given by the expression
\begin{equation}\label{eq:OrderStat}
    f^H_{n,j}(v) = \frac{n!}{(j-1)!(n-j)!}H(v)^{n-j}(1-H(v))^{j-1}h(v).
\end{equation}
We shall also frequently use the following two identities for order statistics:
\begin{align*}
    H(v) f^H_{n-1,j}(v) &= \frac{n-j}{n} f^H_{n,j}(v);\\
    (1-H(v)) f^H_{n-1,j}(v) &= \frac{j}{n} f^H_{n,j+1}(v).
\end{align*}

\begin{definition}[(Weak) First-Order Stochastic (FOS) Dominance]\label{def:fosd}
A distribution $F$ {\em FOS dominates} another distribution $G$, if for every $x$, $F$ gives at least as high a probability of receiving at least $x$ as does $G$: $1-F(x) \ge 1-G(x) \Longleftrightarrow F(x) \le G(x)$.
\end{definition}

\begin{definition}[(Weak) Second-Order Stochastic (SOS) Dominance]\label{def:sosd}
A distribution $F$ {\em SOS dominates} another distribution $G$, if for every $x$: $ \int_{-\infty}^x [F(t) - G(t)] dt \le 0$.
\end{definition}
FOS dominance implies SOS dominance. Another sufficient condition for SOS dominance: $F$ SOS dominates $G$ if $G$ is a mean-preserving spread of $F$.

Let us denote the inner product of two functions $\psi,\phi : [0,1] \rightarrow \bR_+$ by 
\[
    \langle \psi, \phi \rangle = \int_0^1 \psi(x) \phi(x) dx.
\]
When appropriately normalized, $\langle \psi, \phi \rangle$ measures the similarity between the functions $\psi$ and $\phi$.


\subsection{Equilibrium and Objective Function}
We shall study the Bayes--Nash equilibrium of the agents, where the agents in the target group use a symmetric strategy $\alpha(v)$ (symmetric across the agents in the group), where $\alpha(v)$ is the output of an agent with ability $v$; similarly, the agents in the non-target group use a symmetric strategy $\beta(v)$. The objective of the contest designer is to maximize the expected total output generated by the agents in the target group:
\begin{equation*}
    \sum_{i \in [n]} \bE_{v_i \sim F} [\alpha(v_i)] \bP[t_i = T] = n \cdot \mu \cdot \bE_{v_i \sim F} [\alpha(v_i)],
\end{equation*}
which is equivalent to maximizing $\bE_{v \sim F} [\alpha(v)]$ because $n$ and $\mu$ are (fixed) parameters given to the model.

\section{Equilibrium Analysis}\label{sec:eq}
In this section, we characterize the symmetric Bayes--Nash equilibrium (symmetric across the agents in a given group). We first solve for the equilibrium of the contest with both general and target group-specific prizes, then we specialize it for the two special cases: (1) only general prizes; (2) only group-specific prizes.

\subsection{Both Target Group and General Prizes}\label{sec:eqBoth}
In the next theorem, we focus on contests that have a sequence of general prizes $(w_j)_{j \in [n]}$ and a sequence of target group-specific prizes $(\omega_j)_{j \in [n]}$. We will characterize a Bayes--Nash equilibrium where a player in the target group plays an action $\alpha(v)$ as a function of her ability $v \sim F$ and a player in the non-target group plays an action $\beta(v)$ as a function of her ability $v \sim G$. Note that all the players in the target group use the same strategy $\alpha$ and the players in the non-target group use the same strategy $\beta$.
\begin{theorem}\label{thm:eqBoth}
A contest with general prizes $(w_j)_{j \in [n]}$ and target group-specific prizes $(\omega_j)_{j \in [n]}$ has a Bayes--Nash equilibrium where an agent with ability $v \in [0,1]$ uses strategy $\alpha(v)$ if she is in the target group and strategy $\beta(v)$ if she is in the non-target group, where $\alpha(v)$ and $\beta(v)$ are defined as:
\begin{align*}
    \beta(v) &= \int_0^v (\mu f(k(y)) k'(y) + (1-\mu) g(y)) A(y) y d y,\\
    \alpha(v) &= \begin{cases} \beta(k^{-1}(v)), &\text{$0 \le v \le k(1)$} \\
                \beta(1) + \int_{k(1)}^{v} \mu f(y) C(y) y dy, &\text{$k(1) \le v \le 1$}
                \end{cases},
\end{align*}
where:
\begin{itemize}
    \item $k$ is a function defined over $[0,1]$ and is the solution to the following ordinary differential equation, with boundary condition $k(0) = 0$,
    \[
        k'(v)  = \frac{(1-\mu) g(v) (v - k(v))A(v)}{\mu f(k(v)) ( k(v)A(v) + k(v)B(v) - v A(v) )}.
    \]
    \item $A(v) = \sum_{j \in [n]} w_j \psi_j(\mu F(k(v)) + (1-\mu) G(v))$.
    \item $B(v) = \sum_{j \in [n]} \omega_j \psi_j(\mu F(k(v)) + (1-\mu))$.
    \item $C(v) = \sum_{j \in [n]} (w_j + \omega_j) \psi_j(\mu F(v) + (1-\mu))$.
    \item $\psi_j(x) = \binom{n-1}{j-1} x^{n-j-1}(1-x)^{j-2} ( (n-j) - (n-1) x)$.
\end{itemize}
\end{theorem}
Before moving to the proof of Theorem~\ref{thm:eqBoth}, notice that Theorem~\ref{thm:eqBoth} does not provide a closed-form solution for $\alpha(v)$ and $\beta(v)$, it rather provides a characterization using the function $k(v)$ that relates $\alpha(v)$ and $\beta(v)$. Standard equilibrium analysis techniques using first-order equilibrium conditions will give a system of differential equations in $\alpha(v)$ and $\beta(v)$. Using the function $k(v) = \alpha^{-1}(\beta(v))$, we convert this system of differential equations in $\alpha(v)$ and $\beta(v)$ to a comparatively simpler single first-order explicit ordinary differential equation in only one dependent variable $k(v)$. After solving the ordinary differential equation for $k(v)$, we can use $k(v)$ to compute $\alpha(v)$ and $\beta(v)$ as described in the theorem statement. This technique of using the function $k(v)$ is motivated by the work of Amann and Leininger~\shortcite{amann1996asymmetric}, where they use to it analyze two-agent asymmetric contests.

\begin{proof}[Proof of Theorem~\ref{thm:eqBoth}]
Let us assume that $\alpha : [0,1] \rightarrow \bR_+$ and $\beta : [0,1] \rightarrow \bR_+$ are strictly increasing and (almost everywhere) differentiable functions and are the strategies used by the agents in the target group and the non-target group, respectively. We also assume that $\alpha(0) = \beta(0) = 0$, i.e., an agent with no ability doesn't produce any output, and that $\alpha(v) > 0$ and $\beta(v) > 0$ for $v > 0$, i.e., an agent with a non-zero ability produces a non-zero output. Let $\overline{\alpha} = \alpha(1)$ and $\overline{\beta} = \beta(1)$; we shall observe later in the proof that $\alpha(x) \ge \beta(x)$ for all $x$, and therefore, $\overline{\alpha} \ge \overline{\beta}$. Let the distribution of the output produced by an agent in the target group be denoted by $\widehat{F}$; $\widehat{F}$ and the corresponding PDF $\widehat{f}$ are given by:
\[\widehat{F}(x) = F(\alpha^{-1}(x)); \quad \widehat{f}(x) = f(\alpha^{-1}(x)) \frac{d \alpha^{-1}(x)}{dx},\] 
where $x \in [0,\overline{\alpha}]$.
Similarly, the distribution and the PDF of the output produced by a non-target group agent are:
\[\widehat{G}(x) = G(\beta^{-1}(x)); \quad \widehat{g}(x) = g(\beta^{-1}(x)) \frac{d \beta^{-1}(x)}{dx},\]
where $x \in [0,\overline{\beta}]$. Any given agent belongs to the target group with a probability $\mu$, so the overall distribution of the output produced by an agent is given by the following mixture distribution:
\begin{multline*}
    \widehat{H}(x) = \mu \widehat{F}(x) + (1-\mu) \widehat{G}(x) =\\
    \begin{cases} \mu F(\alpha^{-1}(x)) + (1-\mu) G(\beta^{-1}(x)), &\text{$0 \le x \le \overline{\beta}$} \\
    \mu F(\alpha^{-1}(x)) + (1-\mu), &\text{$\overline{\beta} \le x \le \overline{\alpha}$}\end{cases}.
\end{multline*}
For the general prizes $(w_j)_{j \in [n]}$, which are accessible to all agents, the probability that an agent who produces an output of $x$ receives the $j$-th prize is given by $p^{\widehat{H}}_j(x)$, which is equal to
\begin{equation*}
    p^{\widehat{H}}_j(x) = \binom{n-1}{j-1} \widehat{H}(x)^{n-j}(1-\widehat{H}(x))^{j-1}.
\end{equation*}
The target group-specific prizes $(\omega_j)_{j \in [n]}$ are allocated to only the agents in the target group. For these prizes, we can effectively assume that all the agents in the non-target group produce $<0$ output, and then find the probability that an agent in the target group has the $j$-th highest output. This output distribution, denoted by $\tilde{H}$, is given by:
\begin{equation*}
    \tilde{H}(x) = \mu \widehat{F}(x) + (1-\mu) = \mu F(\alpha^{-1}(x)) + (1-\mu), 
\end{equation*}
and the probability that the $x$ is $j$-th highest value is given by $p^{\tilde{H}}_j(x)$. Note that $\tilde{H}$ is not continuous and (its density function) has a point-mass of $(1-\mu)$ at $0$.

\textbf{First-order conditions.} 
We now derive the first-order equilibrium conditions. Let us focus on a particular agent $i$. Let us assume that the agent has an ability $v$ and produces an output $x$. All the other agents are playing with strategy $\alpha$ or $\beta$ depending upon their group. If agent $i$ is in the target group, his expected utility is:
\begin{equation}\label{eq:utility:target}
u_T(v, x) = v \left(\sum_{j \in [n]} w_j p^{\widehat{H}}_j(x)+ \sum_{j \in [n]} \omega_j p^{\tilde{H}}_j(x)\right) - x,
\end{equation}
while, if the agent is not in the target group, then his expected utility is:
\begin{equation}\label{eq:utility:non}
u_N(v, x) = v \sum_{j \in [n]} w_j p^{\widehat{H}}_j(x) - x.
\end{equation}
For a target group agent, the first-order condition implies that the derivative of $u_T(v, x)$ with respect to $x$, at $x = \alpha(v)$, is: $0$ if $x > 0$, and $\le 0$ if $x = 0$. For $x > 0$, we have:
\begin{multline}\label{eq:foc:general:1}
    \left. \frac{d u_T(v,x)}{dx} \right|_{v = \alpha^{-1}(x)} = 0 \implies \\
    \alpha^{-1}(x) \left(\sum_{j \in [n]} w_j \frac{d p^{\widehat{H}}_j(x)}{d x} + \sum_{j \in [n]} \omega_j \frac{d p^{\tilde{H}}_j(x)}{d x} \right) = 1.
\end{multline}
Similarly, for non-target group agent we have:
\begin{multline}\label{eq:foc:general:2}
    \left. \frac{d u_N(v,x)}{dx} \right|_{v = \beta^{-1}(x)} = 0 \implies \\
    \beta^{-1}(x) \sum_{j \in [n]} w_j \frac{d p^{\widehat{H}}_j(x)}{d x} = 1.
\end{multline}
From the two equations above, we can observe that for any $x$, $\alpha^{-1}(x) \le \beta^{-1}(x)$, because the prizes $w_j$'s and $\omega_j$'s and the derivatives of $p^{\widehat{H}}(x)$ and $p^{\tilde{H}}(x)$ are non-negative. $\alpha^{-1}(x) \le \beta^{-1}(x)$ implies $\alpha(v) \ge \beta(v)$ for any given $v$, so an agent in the target group produces more output compared to an agent in the non-target group with the same ability. We also get that $\alpha^{-1}(\beta(v)) \le \beta^{-1}(\beta(v)) = v$.

\textbf{Second-order conditions.}
Note that equation~\eqref{eq:foc:general:1} is valid for all values of $v > 0$ and $x = \alpha(v)$. Let $\xi(x)$ denote
\[
    \xi(x) = \left( \sum_{j \in [n]} w_j \frac{d p^{\widehat{H}}_j(x)}{d x} + \sum_{j \in [n]} \omega_j \frac{d p^{\tilde{H}}_j(x)}{d x} \right)
\]
Select any values $v_1, v, v_2$ such that $v_1 < v < v_2$. Evaluating equation~\eqref{eq:foc:general:1} at $v_1$ and $x_1 = \alpha(v_1)$, we get 
\begin{equation*}
    \left. \frac{d u_T(v_1,x)}{dx} \right|_{x = x_1} = v_1 \xi(x_1) - 1 = 0.
\end{equation*}
Note that $\xi(x)$ non-negative for every $x$ (can be deduced from the formula for $\xi(x)$, also done later in the proof). As $v > v_1$, replacing $v_1$ by $v$ in the equation above, we have
\begin{equation*}
    v \xi(x_1) - 1 >  v_1 \xi(x_1) - 1 = 0 = v \xi(x) - 1.
\end{equation*}
So, $v (\xi(x_1) - \xi(x)) > 0$ if $x_1 < x = \alpha(v)$. Similarly, let $x_2 = \alpha(v_2)$, and as $v < v_2$, we get
\begin{equation*}
    v \xi(x_2) - 1 <  v_2 \xi(x_2) - 1 = 0 = v \xi(x) - 1.
\end{equation*}
So, $v (\xi(x_2) - \xi(x)) < 0$ if $x_2 > x = \alpha(v)$. 
These two equations give us the second-order optimality condition. So, $x = \alpha(v)$ maximizes the utility of a target group agent compared to any other $x > 0$. 

Now, let us focus our attention to $x=0$. At $x = 0$, there is a discontinuity in the win probability, the probability of win is only right-continuous, because the non-target group players are not eligible for the target group-specific prizes, which happens with non-zero probability. But, right-continuity is enough to ensure that for any player with $v>0$, $x = \alpha(v)$ (satisfying the first-order conditions at $x$) is a better strategy for agent $v$ compared to any other $x \ge 0$. And, for a player with $v = 0$, $x=0$ is trivially an optimal strategy.

Similar arguments can be applied for non-target group agents to show that $x = \beta(v)$ is an equilibrium if it satisfies the first-order conditions.


\textbf{Equilibrium strategies.}
We now solve for the equilibrium strategies using the first-order conditions derived earlier. We separate into two cases, depending upon the value of $x$, the output.

\noindent \textbf{Case 1: $x \le \overline{\beta}$.}\\
Let us define a function $k(v)$ over the domain $[0,1]$ as $k(v) = \alpha^{-1}(\beta(v))$. Note that $k(\beta^{-1}(x)) = \alpha^{-1}(x)$, and $k(0)=0$ as per our assumption that $\alpha(0) = \beta(0) = 0$ and $\alpha(v) > 0$ and $\beta(v) > 0$ for $v > 0$. For $x \le \overline{\beta}$, replacing $\alpha^{-1}(x)$ by $k(\beta^{-1}(x))$ in $\widehat{H}(x)$, we get: 
\begin{multline*}
    \widehat{H}(x) = \mu F(\alpha^{-1}(x)) + (1-\mu) G(\beta^{-1}(x)) \\
        = \mu F(k(\beta^{-1}(x))) + (1-\mu) G(\beta^{-1}(x)),
\end{multline*}
which shows that $\widehat{H}(x)$ can be written as function of $k$ and $\beta$, without directly using $\alpha$. The same is true also for functions based on $\widehat{H}$, like $p^{\widehat{H}}_j(x)$ and $\frac{d p^{\widehat{H}}_j(x)}{d x}$. Differentiating $\widehat{H}(x)$ w.r.t. $x$ we get:
\begin{multline*}
    \frac{d \widehat{H}(x)}{d x} = (\mu f(k(\beta^{-1}(x))) k'(\beta^{-1}(x)) \\
        + (1-\mu) g(\beta^{-1}(x))) \frac{d \beta^{-1}(x)}{d x}.
\end{multline*}
Now, we do a change of variable, replacing $x$ by $\beta(v)$, we get:
\begin{equation*}
    \widehat{H}(\beta(v)) = \mu F(k(v)) + (1-\mu) G(v);
\end{equation*}
\begin{multline*}
    \frac{d \widehat{H}(\beta(v))}{d \beta(v)} = \frac{d \widehat{H}(\beta(v))}{d v} \frac{d v}{d \beta(v)} \\
        = (\mu f(k(v)) k'(v) + (1-\mu) g(v)) \frac{d v}{d \beta(v)}.
\end{multline*}
Similarly, we can write (by differentiating $\tilde{H}$ w.r.t. $x$, for $x > 0$):
\begin{align*}
    \tilde{H}(\beta(v)) &= (1-\mu) + \mu F(k(v)); \\
    \frac{d \tilde{H}(\beta(v))}{d \beta(v)} &= \frac{d \tilde{H}(\beta(v))}{d v} \frac{d v}{d \beta(v)} =  \mu f(k(v)) k'(v) \frac{d v}{d \beta(v)}.
\end{align*}

Observe from the equations above that $\widehat{H}(\beta(v))$ and $\tilde{H}(\beta(v))$ are functions of only $v$ and $k(v)$, and $\frac{d \widehat{H}(\beta(v))}{d v}$ and $\frac{d \tilde{H}(\beta(v))}{d v}$ are functions of only $v$, $k(v)$, and $k'(v)$, and do not directly dependent upon $\beta(v)$.

Expanding $\frac{d p^{\widehat{H}}_j(x)}{d x}$ and changing variable $x$ to $\beta(v)$, we get:
\begin{align*}
    \frac{d p^{\widehat{H}}_j(x)}{d x} &= \binom{n-1}{j-1} \widehat{H}(x)^{n-j-1}(1-\widehat{H}(x))^{j-2} \\
        &\qquad ( (n-j)(1-\widehat{H}(x)) - (j-1)\widehat{H}(x) )\frac{d \widehat{H}(x)}{d x}\\
    &= A_j(v) (\mu f(k(v)) k'(v) + (1-\mu) g(v)) \frac{d v}{d \beta(v)},
\end{align*}
where $A_j(v) = \psi_j(\widehat{H}(\beta(v)))$, where $\psi_j(x) = \binom{n-1}{j-1} x^{n-j-1}(1-x)^{j-2} ( (n-j) - (n-1) x)$. $A_j(v)$ is a function of $v$ and $k(v)$.
In a similar manner, we get: 
\begin{equation*}
    \frac{d p^{\tilde{H}}_j(x)}{d x} = B_j(v) \mu f(k(v)) k'(v) \frac{d v}{d \beta(v)},
\end{equation*}
where $B_j(v) = \psi_j(\tilde{H}(\beta(v)))$.

Replacing $\frac{d p^{\widehat{H}}_j(x)}{d x}$ and $\frac{d p^{\tilde{H}}_j(x)}{d x}$ in the first-order condition, equation (\ref{eq:foc:general:1}), and changing $x$ to $\beta(v)$, we get:
\begin{align*}
    \alpha^{-1}&(x) \left(\sum_{j \in [n]} w_j \frac{d p^{\widehat{H}}_j(x)}{d x} + \sum_{j \in [n]} \omega_j \frac{d p^{\tilde{H}}_j(x)}{d x} \right) = 1 \\
    \Longleftrightarrow& \alpha^{-1}(\beta(v)) \\
        &\qquad \left( \sum_{j \in [n]} w_j A_j(v) (\mu f(k(v)) k'(v) + (1-\mu) g(v)) \right. \\
        &\qquad\qquad \left. + \sum_{j \in [n]} \omega_j B_j(v) \mu f(k(v)) k'(v) \right) = \frac{d \beta(v)}{d v} \\
    \Longleftrightarrow&  k(v) k'(v) \mu f(k(v)) \left(\sum_{j \in [n]} w_j A_j(v) + \sum_{j \in [n]} \omega_j B_j(v) \right) \\
    &\qquad + k(v) (1-\mu) g(v) \sum_{j \in [n]} w_j A_j(v) = \frac{d \beta(v)}{d v}.
\end{align*}
Let  $A(v) = \sum_{j \in [n]} w_j A_j(v)$ and $ B(v) = \sum_{j \in [n]} \omega_j B_j(v)$; from equation above
\begin{multline}
    k(v) k'(v) \mu f(k(v)) ( A(v) + B(v) ) + k(v) (1-\mu) g(v) A(v) \\
    = \frac{d \beta(v)}{d v}.
\end{multline}
In a similar manner, from first-order condition (\ref{eq:foc:general:2}):
\begin{equation}\label{eq:foc:general:21}
    v k'(v) \mu f(k(v)) A(v) + v (1-\mu) g(v) A(v) = \frac{d \beta(v)}{d v}.
\end{equation}
Combining these two equations, we get:
\begin{equation}\label{eq:sol:general:1}
    k'(v)  = \frac{(1-\mu) g(v) (v - k(v))A(v)}{\mu f(k(v)) ( k(v)A(v) + k(v)B(v) - v A(v) )}.
\end{equation}
The expression above is a explicit first-order ordinary differential equation and can be calculated analytically based on the distributions $F$ and $G$. As per our assumptions on $F$ and $G$, and with boundary condition for $k(0) = 0$, there exists a unique solution to the equation above. 
After calculating the solution for $k(v)$, we can calculate $\beta(v)$ by using equation (\ref{eq:foc:general:21}):
\begin{equation}\label{eq:sol:general:2}
    \beta(v) = \int_0^v (\mu k'(y) f(k(y)) + (1-\mu) g(y)) A(y) y d y.
\end{equation}
Let us denote the maximum value that $k$ attains by $\overline{k} = k(1)$. Going back to the definition of $k$, $k(v) = \alpha^{-1}(\beta(v)) \implies k^{-1}(v) = \beta^{-1}(\alpha(v))$ for $v \in [0,\overline{k}]$. So, after calculating $k$ and $\beta$ as shown above, we can calculate $\alpha$ as:
\begin{equation}\label{eq:sol:general:3}
    \alpha(v) = \beta(k^{-1}(v)), \text{ for $v \in [0,\overline{k}]$}.
\end{equation}
Note that at $\overline{k}$, $\alpha(\overline{k}) = \beta(k^{-1}(\overline{k})) = \beta(1) = \overline{\beta}$, which is exactly the range we were working with: $0 \le x \le \overline{\beta}$.

\noindent\textbf{Case 2: $\overline{\beta} \le x \le \overline{\alpha}$.}\\
For $x \ge \overline{\beta}$, we have
\begin{align*}
    \widehat{H}(x) &= \mu F(\alpha^{-1}(x)) + (1-\mu) = \tilde{H}(x).
\end{align*}
As $\widehat{H}(x) = \tilde{H}(x)$ for $x \ge \overline{\beta}$, we also have $p_j^{\widehat{H}}(x) = p_j^{\tilde{H}}(x)$. As we did earlier, we do a change of variable from $x$ to $\alpha(v)$ to get
\begin{align*}
    \widehat{H}(\alpha(v)) &= \mu F(v) + (1-\mu); \ \frac{d\widehat{H}(\alpha(v))}{d \alpha(v)} = \mu f(v) \frac{d v}{d \alpha(v)}.
\end{align*}
Differentiating $p^{\widehat{H}}_j(x)$ w.r.t. $x$ and changing variable to $\alpha(v)$, we get
\begin{align*}
    \frac{d p^{\widehat{H}}_j(x)}{d x} &= \binom{n-1}{j-1} \widehat{H}(x)^{n-j-1}(1-\widehat{H}(x))^{j-2} \\
        &\qquad ( (n-j)(1-\widehat{H}(x)) - (j-1)\widehat{H}(x) )\frac{d \widehat{H}(x)}{d x}\\
    &= C_j(v) \mu f(v) \frac{d v}{d \alpha(v)},
\end{align*}
where $C_j(v) = \psi_j(\tilde{H}(\alpha(v)))$.

From the first order conditions, we get (note that $\widehat{H}(x) = \tilde{H}(x)$)
\begin{align*}
    &\quad \alpha^{-1}(x) \left(\sum_{j \in [n]} w_j \frac{d p^{\widehat{H}}_j(x)}{d x} + \sum_{j \in [n]} \omega_j \frac{d p^{\tilde{H}}_j(x)}{d x} \right) = 1 \\
    &\Longleftrightarrow v \sum_{j \in [n]} (w_j + \omega_j) C_j(v) \mu f(v) = \frac{d \alpha(v)}{d v}.
\end{align*}
Using this differential equation, along with the boundary condition $\alpha(\overline{k}) = \overline{\beta}$ that we derived earlier, we get
\begin{equation}
    \alpha(v) = \overline{\beta} + \int_{\overline{k}}^{v} y \sum_{j \in [n]} (w_j + \omega_j) C_j(y) \mu f(y) dy.
\end{equation}

\end{proof}

If the distributions $F$ or $G$, or the prizes $(w_j)_{j \in [n]}$ or $(\omega_j)_{j \in [n]}$, are non-trivial, then it may be intractable to calculate the analytical solutions for $\alpha(v)$ and $\beta(v)$ derived in Theorem~\ref{thm:eqBoth}, but numerical solutions may be computed. 
Example~\ref{ex:eq:1} in the appendix illustrates the use of Theorem~\ref{thm:eqBoth} to compute $\alpha(v)$ and $\beta(v)$ for a particular instantiation of the problem.

Theorem~\ref{thm:eqBoth} characterizes an equilibrium assuming that all players in a given group use the same strategy ($\alpha(v)$ or $\beta(v)$), and that $\alpha(v)$ and $\beta(v)$ are strictly increasing and (almost everywhere) differentiable functions of $v$. We believe that these assumptions are reasonable and the equilibrium is the most natural one for this model. It remains open to prove the uniqueness of this equilibrium or to characterize the set of all possible equilibria.

\subsection{General Prizes Only}\label{sec:eqGen}
A contest designer may want to only award prizes that are equally accessible to all the agents. This model is particularly appealing for several real-life applications, because, although the design of the prize structure may incorporate distributional information about the ability of the agents in the two groups, the contest itself is completely unbiased. 

In our model, this design choice corresponds to setting $\omega_j = 0$ for all $j \in [n]$. For this case, the expected utility for a target group agent and a non-target group agent for a given pair of ability $v$ and output $b$ is the same. We can specialize the equilibrium characterization in Theorem~\ref{thm:eqBoth} to this case as follows:
\begin{theorem}\label{thm:eqGen}
A contest with only general prizes $(w_j)_{j \in [n]}$ has a Bayes--Nash equilibrium where an agent with ability $v \in [0,1]$ in the target or non-target group uses the strategy $\alpha(v)$, where $\alpha(v)$ is defined as:
\begin{equation*}
    \alpha(v) = \sum_{j \in [n-1]} (w_j - w_{j+1}) \int_0^v y f^{\mu F + (1-\mu) G}_{n-1,j}(y) dy.
\end{equation*}
\end{theorem}
\begin{proof}
This result is a specific instantiation of Theorem~\ref{thm:eqBoth}. Note that the utility functions are the same for both target group and non-target group agents, $u_T(v,x) = u_N(v,x)$ as given in equations~\eqref{eq:utility:target} and \eqref{eq:utility:non}. With reference to Theorem~\ref{thm:eqBoth} and its proof, we have $\alpha(v) = \beta(v)$ and $k(v) = \alpha^{-1}(\beta(v)) = v$, and therefore
\begin{multline*}
    \alpha(v) = \int_0^v y \sum_{j \in [n]} w_j \psi_j(\mu F(y) + (1-\mu) G(y)) \\
            (\mu f(y) + (1-\mu) g(y)) d y,
\end{multline*}
where $\psi_j(x) = \binom{n-1}{j-1} x^{n-j-1}(1-x)^{j-2} ( (n-j) - (n-1) x)$, as before.

Let $H(v) = \mu F(v) + (1-\mu) G(v)$. We have $\alpha(v) = $\\
$\int_0^v y \sum_{j \in [n]} w_j \psi_j(H(y)) h(y) d y$. Rewriting $\alpha(v)$ using the PDF for the order statistic of $H$,  $f^H_{n,j}$ (formula given in equation \eqref{eq:OrderStat}), we get
\[
    \alpha(v) = \sum_{j \in [n-1]} (w_j - w_{j+1}) \int_0^v y f^H_{n-1,j}(y) dy,
\]
as required.
\end{proof}

This equilibrium is unique, as follows from a result by Chawla and Hartline~\shortcite{chawla2013auctions}. They prove that if a contest is anonymous (does not discriminate among the agents) and the abilities of the agents are sampled i.i.d. from a distribution, then the equilibrium is unique. With only general prizes, the contest is (i) anonymous because the general prizes are awarded without any bias towards players from any group, and (ii) effectively, the abilities of the agents are sampled i.i.d. from the distribution $\mu F(v) + (1-\mu) G(v)$.

\subsection{Group-Specific Prizes Only}\label{sec:eqTarget}
In this case, there are no general prizes accessible to agents from both groups: all the prizes are allocated based on strict quotas for the groups. This model is similar to the model of seat quotas for college admissions studied by Bodoh-Creed and Hickman~\shortcite{bodoh2018college}. Technically, in this case we assume that $w_j = 0$ for all $j \in [n]$. Here we focus on the target group; a similar result also holds for the non-target group, if there are any prizes reserved for them.
\begin{theorem}\label{thm:eqTarget}
A contest with target only group-specific prizes, $(\omega_j)_{j \in [n]}$, has a Bayes--Nash equilibrium where an agent with ability $v \in [0,1]$ in the target group uses strategy $\alpha(v)$, where $\alpha(v)$ is defined as:
\begin{equation*}
    \alpha(v) = \sum_{j \in [n-1]} (\omega_j - \omega_{j+1}) \int_0^v y f^{\mu F + (1-\mu)}_{n-1,j}(y) dy.
\end{equation*}
\end{theorem}
\begin{proof}
Substituting $w_j = 0$ for every $j \in [n]$ in Theorem~\ref{thm:eqBoth}, we get $\beta(v) = 0$, $k(v) = 0$, and $C(v) = \sum_{j \in [n]} \omega_j \psi_j(\mu F(v) + (1-\mu))$ for $v \in [0,1]$, and therefore,
\[
    \alpha(v) = \int_{0}^{v} \mu f(y) C(y) y dy.
\]
Let $H(v) = \mu F(v) + (1-\mu)$. We have $\alpha(v) = $ \\
$\int_0^v y \sum_{j \in [n]} \omega_j \psi_j(H(y)) h(y) d y$. Rewriting $\alpha(v)$ using the PDF for the order statistic of $H$, $f^H_{n,j}$, we have
\[
    \alpha(v) = \sum_{j \in [n-1]} (\omega_j - \omega_{j+1}) \int_0^v y f^H_{n-1,j}(y) dy,
\]
as required.
\end{proof}

Chawla and Hartline~\shortcite{chawla2013auctions}'s uniqueness result also applies to this equilibrium characterization, as it did for only general prizes (Section~\ref{sec:eqGen}). We can transform 
an instance of our problem (without changing its set of equilibria) to satisfy these requirements: the abilities of the players are picked i.i.d. from the distribution $\mu F(v) + (1-\mu)$ and the contest awards prizes without discriminating among the players. 

\section{Designing Prizes to Incentivize the Target Group}\label{sec:des}
Based on the equilibrium characterizations in the previous section, in this section we investigate the properties of the optimal contest that maximizes the total output of the target group. We characterize the optimal contests for only general prizes and for only group-specific prizes. Characterization of the optimal contest for both general and group-specific prizes is a non-trivial open problem because of the absence of a closed-form equilibrium characterization for the case (the characterization in Theorem~\ref{thm:eqBoth} involves an ordinary differential equation).



In the appendix (Section~\ref{sec:desCompare}), 
we compare the choice between only general and only group-specific prizes; we observe that either of them may be a better choice depending upon the situation. Even if the distributions of the target and the non-target group are the same, the relative frequency of target group agents, measured by $\mu$, determines which prize allocation scheme is better. If $\mu$ is sufficiently high, then it is better to have only group-specific prizes, while the converse is true if $\mu$ is very low.

\subsection{General Prizes Only}\label{sec:desGen}
The strategy used by the agents in the target group is $\alpha(v)$, as derived in Theorem~\ref{thm:eqGen}. It can be rewritten as
\begin{align}
    \alpha(v) &= \sum_{j \in [n-1]} (w_j - w_{j+1}) \int_0^v y f^{\mu F + (1-\mu) G}_{n-1,j}(y) dy \nonumber \\
    &= \sum_{j \in [n-1]} \gamma_j \frac{1}{j} \int_0^v y f^{\mu F + (1-\mu) G}_{n-1,j}(y) dy,
\end{align}
where $\gamma_j = j (w_j - w_{j+1})$. Instead of optimizing over $(w_j)_{j \in [n]}$, we can optimize over $(\gamma_j)_{j \in [n-1]}$ with constraints $\gamma_j \ge 0$ and $\sum_j \gamma_j = 1$, to find the optimal contest that maximizes the target group's output.

\begin{theorem}\label{thm:desGen}
The contest with only general prizes that maximizes the expected total output of the agents in the target group awards a prize of $1/k^*$ to the $k^*$ top-ranked agents and a prize of $0$ to the remaining $n-k^*$ agents, where $k^*$ is defined as
\[
    k^* \in \argmax_{j \in [n-1]} \langle \psi , f^{\mu F + (1-\mu) G}_{n,j+1} \rangle
\]
for $\psi(x) = \frac{x (1-F(x))}{1- \mu F(x) - (1-\mu) G(x)}$.
\end{theorem}
\begin{proof}
The expected output of an agent in the target group is $\bE_{v \sim F}[\alpha(v)]$, and maximizing this quantity maximizes the total output produced by all agents in the target group because $\mu$ and $n$ are constants. Let $H(v) = \mu F(v) + (1-\mu) G(v)$.
\begin{align*}
    \bE_{v \sim F}[\alpha(v)] &= \int_0^1 \alpha(v) dF(v) \\
    &= \int_0^1 \left( \sum_{j \in [n-1]} \gamma_j \frac{1}{j} \int_0^v y f^{H}_{n-1,j}(y) dy \right) dF(v)\\
    &= \sum_{j \in [n-1]} \gamma_j \frac{1}{j} \int_0^1 \left( \int_y^1 dF(v) \right) y f^{H}_{n-1,j}(y) dy \\ 
    &= \sum_{j \in [n-1]} \gamma_j \frac{1}{j} \int_0^1 y (1-F(y)) f^{H}_{n-1,j}(y) dy. 
\end{align*}
As $\gamma_j \ge 0$ and $\sum_j \gamma_j = 1$, from the above expression we can observe that the optimal solution sets $\gamma_k > 0$ only if $\left(\frac{1}{j} \int_0^1 y (1-F(y)) f^{H}_{n-1,j}(y) dy \right)$ is maximized by $j = k$. Moreover, there is always an optimal solution where $\gamma_{k^*} = 1$ for some $k^* \in [n-1]$ and $\gamma_j = 0$ for $j \neq k^*$, where $k^*$ is given by:
\begin{align*}
    k^* &\in \argmax_{j \in [n-1]} \frac{1}{j} \int_0^1 y (1-F(y)) f^{H}_{n-1,j}(y) dy \\
    &= \argmax_{j \in [n-1]} \frac{1}{n} \int_0^1 y \frac{1-F(y)}{1-H(y)} f^{H}_{n,j+1}(y) dy \\
    &= \argmax_{j \in [n-1]} \langle \psi , f^{H}_{n,j+1} \rangle,
\end{align*}
where $\psi(x) = x \frac{1-F(x)}{1-H(x)}$.
\end{proof}

Let $H(x)$ denote $\mu F(x) + (1-\mu) G(x)$. The function $\psi(x) = x\frac{1-F(x)}{1 - H(x)}$ in Theorem~\ref{thm:desGen} is independent of the prize structure: it depends only on $\mu$, $F$, and $G$, which are parameters given to the contest designer. The function $f^{H}_{n,j+1}(x)$ is the PDF of the $(j+1)$-th highest order statistic of the distribution $H$. The optimal $j = k^*$ maximizes the inner product between $\psi$ and $f^H_{n,j+1}$, and therefore, the similarity between them. In other words, the optimal $j = k^*$ selects the order statistic that is as similar to $\psi$ as possible.

Note that when we are maximizing the total output, then instead of an inner product of the order statistic with $\psi(x)$ in Theorem~\ref{thm:desGen}, we take an inner product with just $x$. As $x$ is monotonically increasing, the order statistic that is most similar to $x$ and that maximizes the inner product with $x$ is the one with $j = 1$, as shown by Moldovanu and Sela~\shortcite{moldovanu2001optimal}. This provides an argument in favor of allocating the entire prize budget to the first prize. In contrast, when we focus on the target group only, it may be optimal to distribute the prize among several top-ranked agents. This also means that we may lose a portion of the total output, and this loss can be $\Omega(n)$ as shown in Example~\ref{ex:des:0}. Also, by the \textit{single-crossing} property, see Vojnovic~\cite{vojnovic2015contest} Chapter 3, if we flatten the prize structure then we increase the output of the agents with low ability and decrease the output of the agents with high ability.

We have partially omitted calculations from the examples in this section. 
These examples are repeated in the appendix with more details. 

\begin{example}\label{ex:des:0}
Let $F(x) = 1 - (1-x)^{n-1}$, $G(x) = ((n-1)x - F(x))/(n-2)$, and $\mu = 1/(n-1)$. We get $H(x) = \mu F(x) + (1-\mu) G(x) = x$, which is the uniform distribution over $[0,1]$. Note that 
$\psi(x) = \frac{x(1-F(x))}{1-H(x)} = x(1-x)^{n-2}$.

It can be checked that the order-statistic that maximizes the inner-product with $\psi(x)$ is $f^H_{n,n-1}(x) = n(n-1)x(1-x)^{n-1}$. So, the optimal value is $k^* = n-2$. The total output generated for $k^* = n-2$ is $\frac{2}{n(n+1)}$.
On the other hand, we know that the expected total output is maximized by $j=1$ and is equal to $\frac{(n-1)}{n(n+1)}$.
The ratio between the two quantities is $(n-1)/2 = \Omega(n)$.
\end{example}

In the above example, the distribution of the target group $F$ was first-order stochastic (FOS) dominated by the distribution of the general population $H$.
For this case, we observed that awarding prizes to more than one agent increases the expected output of the target group. $F$ being FOS dominated by $H$ means that the agents in the target group have lower ability than the general population. As flattening the prize structure increases the output of the lower ability agents and decreases the output of the higher ability agents, it makes sense that having a flatter prize structure incentivizes them.

The above result raises the question: What if the target group is equally able or stronger than the general population, i.e., if the distribution of the ability of an agent in the target group, $F$, FOS dominates the distribution of the ability of an agent in the non-target group, $G$; is it then optimal to award the prize to the top-ranked agent only? When $F = G = H$, we know that it is optimal to give prize to the top-ranked agent only, so when $F$ FOS dominates $H$, one might expect this to be the case as well. However, the example below illustrates that this is not true.
\footnote{In the preliminaries we assumed $F$ to be continuous, strictly increasing, and differentiable in $[0,1]$. In the following and the subsequent examples, the distribution $F$ is continuous and weakly increasing but may not be strictly increasing and differentiable everywhere. If $F$ is not strictly increasing, we can add a slight gradient to resolve the issue; on the other hand, we can smooth out the finite number of points where it is not differentiable. These changes will have a minimal effect on the objective value and our analysis holds. For ease of presentation, we shall not discuss these issues.}

\begin{example}\label{ex:des:2}
Let $\mu = 1/8$, $F(x) = 1 - S(x)$, $G(x) = (x-\mu F(x))/(1-\mu)$, where $S(x)$ is given by:
\[
    S(x) = \begin{cases} 1, &\text{ if $x < 3/4$} \\
    16(1-x)^2, &\text{ if $3/4 \le x < 15/16$} \\
    1-x, &\text{ if $x \ge 15/16$}
    \end{cases}.
\]
It can be verified that $F(x)$ and $G(x)$ are continuous cumulative distribution functions. The distribution of the general population is $H(x) = \mu F(x) + (1-\mu) G(x) = x$. Also, $F(x) \le H(x)$ for every $x$ (and strict for some values of $x$), and therefore, $F$ FOS dominates $H$. The optimal number of prizes $k^*$ need not be $1$; for example, for $n=50$, $k^* = 11$. 
\end{example}

Now we consider the case where $F$ and $H$ have the same mean but different variance, which implies second-order stochastic (SOS) dominance. The distribution with lower variance SOS dominates the one with higher variance, assuming both distributions have the same mean. The results are similar to FOS dominance: it may be optimal to have multiple prizes irrespective of whether $F$ or $H$ has a higher variance, as shown in the following examples. Note that FOS dominance implies SOS dominance, but not vice-versa.

\begin{example}\label{ex:des:3}
Let $\mu = 2/3$, $F(x) = 3x^2 - 2x^3$, and $G(x) = 3x - 6x^2 + 4x^3$. Observe that $H(x) = \mu F(x) + (1-\mu) G(x) = x$, and the variance of $H$ is $1/12$ while the variance of $F$ is $1/20$, and both have mean $1/2$.

Solving for $k^*$, we get $k^* = \frac{5n}{6} - \frac{\sqrt{7n^2+30n+39}}{6}+\frac{1}{2}$. For example, for $n=50$, $k^* = 19$.  
\end{example}

\begin{example}\label{ex:des:4}
 Let $\mu = 1/4$, $F(x) = 1-S(x)$, and $G(x) = (x-\mu F(x))/(1-\mu)$, where $S(x)$ is given by:
\[
    S(x) = \begin{cases} 1-48x/31, &\text{ if $x < 31/96$} \\
    1/2, &\text{ if $31/96 \le x < 3/4$} \\
    8(1-x)^2, &\text{ if $3/4 \le x < 7/8$} \\
    1-x, &\text{ if $x \ge 7/8$}
    \end{cases}
\]
It can be checked that $F(x)$ and $G(x)$ are valid and for the distribution of the general population we have $H(x) = \mu F(x) + (1-\mu) G(x) = x$. It can also be checked that the mean of all the distributions is $1/2$, and the variance of $H$ is $1/12 \approx 0.083$ while the variance of $F$ is $6703/55296 \approx 0.121 > 1/12$. The optimal number of prizes $k^*$ need not be $1$; for example, for $n=50$, $k^* = 11$. 
\end{example}

\subsection{Group-Specific Prizes Only}\label{sec:desTarget}
In this section, we study the optimal prize structure when the prizes are accessible to the target group only. This models situations when there are separate contests for the target and the non-target group (there can be a separate contest for only non-target group agents, this will not have any affect on the strategy of the target group agents). We see that the optimal contest allocates the entire prize budget to the first prize, i.e., gives a prize of $1$ to the top-ranked agent in the group (Theorem~\ref{thm:desTarget}). The proof of Theorem~\ref{thm:desTarget} extends the techniques used to provide a similar result for total output~\cite{moldovanu2001optimal}.

\begin{theorem}\label{thm:desTarget}
The contest with only group-specific prizes that maximizes the output of the target group awards a prize of $1$ to the top-ranked agent and $0$ to others.
\end{theorem}
\begin{proof}
From the equilibrium characterization in Theorem~\ref{thm:eqTarget}, we have
\begin{align*}
    \alpha(v) &= \sum_{j \in [n-1]} (\omega_j - \omega_{j+1}) \int_0^v y f_{n-1,j}^{\mu F + (1-\mu)}(y) dy \\
    &= \sum_{j \in [n-1]} \gamma_j \frac{1}{j} \int_0^v y f_{n-1,j}^{\mu F + (1-\mu)}(y) dy,
\end{align*}
where $\gamma_j = j(\omega_j - \omega_{j+1}) \ge 0$. 

Let $H(x) = (1-\mu) + \mu F(x)$. We have $h(x) = \mu f(x)$ for $x > 0$ and $(1-H(x)) = \mu(1-F(x))$. Let 
$$
H_{n,j}(x) = \sum_{\ell = 0}^{j-1} \binom{n}{\ell} H(x)^{n-\ell} (1-H(x))^{\ell}
$$
be the distribution of the $j$-th highest order statistic of $H(x)$, and let $h_{n,j}(x)$ be its PDF, which is well-defined for $x > 0$, but at $x = 0$, $H_{n,j}(0)$ has a non-zero point mass. Note that $f_{n,j}^{\mu F + (1-\mu)}(x) = h_{n,j}(x)$.

The objective of the designer is $\bE_{v \sim F}[\alpha(v)]$
\begin{align*}
    &= \int_0^1 \left( \sum_{j \in [n-1]} \gamma_j \frac{1}{j} \int_0^v y h_{n-1,j}(y) dy \right) dF(v)
\end{align*}
\begin{align*}
    &= \sum_{j \in [n-1]} \gamma_j \frac{1}{j}  \int_0^1 y (1-F(y)) h_{n-1,j}(y) dy\\
    &= \sum_{j \in [n-1]} \gamma_j \frac{1}{\mu j}  \int_0^1 y (1-H(y)) h_{n-1,j}(y) dy \\ 
    &= \sum_{j \in [n-1]} \gamma_j \frac{1}{\mu n}  \int_0^1 y h_{n,j+1}(y) dy. 
\end{align*}
As $1/(\mu n)$ is a constant, to maximize the objective, the designer sets $\gamma_j > 0$ only for values of $j$ that maximize $\int_0^1 y h_{n,j+1}(y) dy$. From the definition of $H_{n,k}$, we have 
\begin{multline*}
    1 = H_{n,j+1}(1) = H_{n,j+1}(0) + \int_0^1 h_{n,j+1}(y) dy \\
        \implies \int_0^1 h_{n,j+1}(y) dy = 1 - H_{n,j+1}(0).
\end{multline*}
We also know that for any $x$ and $j < n$,
\begin{multline*}
    H_{n,j}(x) = \sum_{\ell = 0}^{j-1} \binom{n}{\ell} H(x)^{n-\ell} (1-H(x))^{\ell} \\
    \le \sum_{\ell = 0}^{j} \binom{n}{\ell} H(x)^{n-\ell} (1-H(x))^{\ell} = H_{n,j+1}(x).
\end{multline*}
Using this, we get
\begin{align}\label{eq:design:group:1}
    H_{n,1+1}(x) &\le H_{n,j+1}(x) \nonumber \\
    &\implies 1 - H_{n,2}(0) \ge 1 - H_{n,j+1}(0) \nonumber \\
    &\implies \int_0^1 h_{n,2}(y) dy \ge \int_0^1 h_{n,j+1}(y) dy,
\end{align}
for any $j \in [n-1]$. Observe that $h_{n,2}(y)$ is either \textit{single-crossing} with respect to $h_{n,j+1}(y)$, i.e., there is a point $x \in [0,1]$ s.t. $h_{n,2}(y) \le h_{n,j+1}(y)$ for $y \le x$ and $h_{n,2}(y) > f^H_{n,j+1}(y)$ for $y > x$, or $h_{n,2}(y)$ is always greater than $h_{n,j+1}(y)$, because
\begin{align*}
    \frac{h_{n,2}(y)}{h_{n,j+1}(y)} &= \frac{\binom{n-1}{1} H(y)^{n-2} (1-H(y))^{1} }{\binom{n-1}{j} H(y)^{n-j-1} (1-H(y))^{j}} \\
    &= \frac{(n-j-1)! j!}{(n-2)! (1-H(y))^{j-1}} \left( \frac{H(y)}{1-H(y)} \right)^{j-1},
\end{align*}
$\frac{H(y)}{1-H(y)}$ is strictly increasing and goes from $\frac{H(0)}{1-H(0)}$ to $\infty$ as $y$ goes to $1$, and the coefficient of $\left( \frac{H(y)}{1-H(y)} \right)^{j-1}$ is a constant. Using: (i) the identity function $y$ is an increasing function, (ii) $h_{n,2}(y)$ is single-crossing w.r.t. $h_{n,j+1}(y)$, and (iii) equation~(\ref{eq:design:group:1}), we have 
\begin{equation*}
    \int_0^1 y h_{n,2}(y) dy \ge \int_0^1 y h_{n,j+1}(y) dy.
\end{equation*}
So, it is optimal to only have a first prize.
\end{proof}
\section{Conclusion}
Our paper studies rank-order allocation prizes that aim to maximize the output of a target group. We provide equilibrium characterization for these contests and study the properties of the optimal contest.
For unbiased contests (i.e., with general prizes only), although it is preferable to award the prize to the top-performing player if the target group players are \textit{similar} to the overall population, if the target group players are very \textit{different}, then we may want to flatten out the prize structure. This does not only mean that we flatten out the prizes if the target group players are \textit{weaker} than others; we also demonstrate, if the target group players are \textit{stronger} but quite different than others, then also it may be a good idea to flatten out the prizes.
But for contests with fixed prize quotas (i.e., with group-specific prizes only), it may not be useful to flatten the prizes.
Regarding the choice between unbiased contests and fixed quota contests, the size of the target group in proportion to the size of the population plays an important role.

An important open problem is to understand the properties of the optimal contest that has both group-specific and general prizes. We characterize the equilibrium for this problem and analyze the optimal contest for specific sub-cases, but the general problem is open. One way to make this tractable may be to assume large numbers of agents and prizes; such assumptions allow for a stronger and more tractable equilibrium characterization by making the influence of an individual agent infinitesimal~\cite{olszewski2016large}. 
Our work focused on contests with a rank-order allocation of prizes with incomplete information and linear cost functions. 
Instead of rank-order allocation of prizes, we can study a proportional allocation of prizes or its generalizations (see, e.g., \cite{tullock1980efficient}); one motivation for proportional allocation is the randomness and unpredictability of outcome in the real world. Similarly, we may relax the assumption of linear cost functions, or we can study complete information models instead of incomplete information.

\subsection*{Ethical Considerations}
There may be both short-term and long-term implications for designing contests specifically to incentivize a target group. For example, even if the contest itself is unbiased (Section \ref{sec:desGen}), optimizing the expected total output of the players in the target group may mean that we decrease the expected total output of the entire population, and in particular decrease the expected total output of the non-target group. Moreover, such effects can be stronger if the target group forms a smaller fraction of the participants. Our analysis and examples try to capture such effects. For any given contest, this means that we are optimizing for the target group with a cost to society. But for a longer time frame, the answer is less obvious, and we believe that it is out of the scope of our paper. This question concerns whether practices like affirmative action have an overall positive or negative impact on our society in the long run. Our paper addresses the question: if we do want to use affirmative action in contests to incentivize the target group (with particular design choices, like the prize allocation scheme is unbiased), how to do it most effectively, and how it affects the output of all the players in the given contest.


\bibliographystyle{named}
\bibliography{ref}

\begin{thebibliography}{}

\bibitem[\protect\citeauthoryear{Amann and
  Leininger}{1996}]{amann1996asymmetric}
Erwin Amann and Wolfgang Leininger.
\newblock Asymmetric all-pay auctions with incomplete information: the
  two-player case.
\newblock {\em Games and economic behavior}, 14(1):1--18, 1996.

\bibitem[\protect\citeauthoryear{Archak and
  Sundararajan}{2009}]{archak2009optimal}
Nikolay Archak and Arun Sundararajan.
\newblock Optimal design of crowdsourcing contests.
\newblock In {\em Proceedings of 13th International Conference on Information
  Systems}, page 200, Atlanta, USA, 01 2009. AIS.

\bibitem[\protect\citeauthoryear{Barut and Kovenock}{1998}]{barut1998symmetric}
Yasar Barut and Dan Kovenock.
\newblock The symmetric multiple prize all-pay auction with complete
  information.
\newblock {\em European Journal of Political Economy}, 14(4):627--644, 1998.

\bibitem[\protect\citeauthoryear{Baye \bgroup \em et al.\egroup
  }{1996}]{baye1996all}
Michael~R Baye, Dan Kovenock, and Casper~G De~Vries.
\newblock The all-pay auction with complete information.
\newblock {\em Economic Theory}, 8(2):291--305, 1996.

\bibitem[\protect\citeauthoryear{Bodoh-Creed and
  Hickman}{2018}]{bodoh2018college}
Aaron~L Bodoh-Creed and Brent~R Hickman.
\newblock College assignment as a large contest.
\newblock {\em Journal of Economic Theory}, 175:88--126, 2018.

\bibitem[\protect\citeauthoryear{Chawla and
  Hartline}{2013}]{chawla2013auctions}
Shuchi Chawla and Jason~D Hartline.
\newblock Auctions with unique equilibria.
\newblock In {\em Proceedings of the 14th ACM Conference on Electronic
  Commerce}, pages 181--196, New York, USA, 2013. ACM.

\bibitem[\protect\citeauthoryear{Chawla \bgroup \em et al.\egroup
  }{2019}]{chawla2019optimal}
Shuchi Chawla, Jason~D Hartline, and Balasubramanian Sivan.
\newblock Optimal crowdsourcing contests.
\newblock {\em Games and Economic Behavior}, 113:80--96, 2019.

\bibitem[\protect\citeauthoryear{Chowdhury \bgroup \em et al.\egroup
  }{2020}]{chowdhury2020heterogeneity}
Subhasish~M Chowdhury, Patricia Esteve-Gonz{\'a}lez, and Anwesha Mukherjee.
\newblock Heterogeneity, leveling the playing field, and affirmative action in
  contests.
\newblock SSRN:3655727, 2020.

\bibitem[\protect\citeauthoryear{Deng \bgroup \em et al.\egroup
  }{2021}]{deng2021optimally}
Shanglyu Deng, Qiang Fu, and Zenan Wu.
\newblock Optimally biased tullock contests.
\newblock {\em Journal of Mathematical Economics}, 92:10--21, 2021.

\bibitem[\protect\citeauthoryear{Elkind \bgroup \em et al.\egroup
  }{2021}]{elkind2021contest}
Edith Elkind, Abheek Ghosh, and Paul Goldberg.
\newblock Contest design with threshold objectives.
\newblock In {\em Proceedings of the 17th Conference on Web and Internet
  Economics}, page 554, Potsdam, Germany, 12 2021. Springer.

\bibitem[\protect\citeauthoryear{Fu and Wu}{2020}]{fu2020optimal}
Qiang Fu and Zenan Wu.
\newblock On the optimal design of biased contests.
\newblock {\em Theoretical Economics}, 15(4):1435--1470, 2020.

\bibitem[\protect\citeauthoryear{Gavious and
  Minchuk}{2014}]{gavious2014revenue}
Arieh Gavious and Yizhaq Minchuk.
\newblock Revenue in contests with many participants.
\newblock {\em Operations Research Letters}, 42(2):119--122, 2014.

\bibitem[\protect\citeauthoryear{Kirkegaard}{2012}]{kirkegaard2012favoritism}
Ren{\'e} Kirkegaard.
\newblock Favoritism in asymmetric contests: Head starts and handicaps.
\newblock {\em Games and Economic Behavior}, 76(1):226--248, 2012.

\bibitem[\protect\citeauthoryear{Moldovanu and
  Sela}{2001}]{moldovanu2001optimal}
Benny Moldovanu and Aner Sela.
\newblock The optimal allocation of prizes in contests.
\newblock {\em American Economic Review}, 91(3):542--558, 2001.

\bibitem[\protect\citeauthoryear{Olszewski and
  Siegel}{2016}]{olszewski2016large}
Wojciech Olszewski and Ron Siegel.
\newblock Large contests.
\newblock {\em Econometrica}, 84(2):835--854, 2016.

\bibitem[\protect\citeauthoryear{Olszewski and
  Siegel}{2020}]{olszewski2020performance}
Wojciech Olszewski and Ron Siegel.
\newblock Performance-maximizing large contests.
\newblock {\em Theoretical Economics}, 15(1):57--88, 2020.

\bibitem[\protect\citeauthoryear{Siegel}{2009}]{siegel2009all}
Ron Siegel.
\newblock All-pay contests.
\newblock {\em Econometrica}, 77(1):71--92, 2009.

\bibitem[\protect\citeauthoryear{Tullock}{1980}]{tullock1980efficient}
Gordon Tullock.
\newblock Efficient rent-seeking.
\newblock In James Buchanan, Robert Tollison, and Gordon Tullock, editors, {\em
  Toward a Theory of the Rent Seeking Society}, pages 131--146. Texas A \& M
  University, College Station, Texas, USA, 1980.

\bibitem[\protect\citeauthoryear{Vojnović}{2016}]{vojnovic2015contest}
Milan Vojnović.
\newblock {\em Contest Theory: Incentive Mechanisms and Ranking Methods}.
\newblock Cambridge University Press, Cambridge, UK, 2016.

\end{thebibliography}

\appendix


\section{Examples}
\subsection{Equilibrium Analysis}
The following example illustrates the use of Theorem~\ref{thm:eqBoth} to compute $\alpha(v)$ and $\beta(v)$.
\begin{example}\label{ex:eq:1}
    Let $\mu = 1/2$, $F(v) = v^s$, and $G(v) = v^t$, for $s,t > 0$. Let $w_1 = 1/2$, $w_j = 0$ for $j > 1$, $\omega_1 = 1/2$, and $\omega_j = 0$ for $j > 1$. 
    
    We get $\mu F(k(v)) + (1-\mu) G(v) = (k(v)^s + v^t)/2$, $\mu F(k(v)) + (1-\mu) = (k(v)^s + 1)/2$, and $\mu F(v) + (1-\mu) = (v^s + 1)/2$. Further, we get
    \begin{align*}
        A(v) &= \sum_{j \in [n]} w_j \psi_j(\mu F(k(v)) + (1-\mu) G(v)) \\
            &= w_1 \psi_1(\mu F(k(v)) + (1-\mu) G(v)) \\
            &= \frac{1}{2} (n-1) \left( \frac{k(v)^s + v^t}{2} \right)^{n-2}.
    \end{align*}
    Similarly, we get $B(v) = \frac{1}{2} (n-1) \left( \frac{k(v)^s + 1}{2} \right)^{n-2}$ and $C(v) = (n-1) \left( \frac{v^s + 1}{2} \right)^{n-2}$. Putting it all together, we get
    \begin{equation*}
        k'(v) = \frac{t v^{t-1} (v-k(v)) (k(v)^s+v^t)^{n-2}}{%
        \splitfrac{s k(v)^{s-1} (k(v) (k(v)^s+1)^{n-2}}
        { - (v-k(v)) (k(v)^s+v^t)^{n-2})}}.
    \end{equation*}
    Numerical solution for $k(v)$ for $n=5$, $s=1/2$, and $t=1$ is given in Figure~\ref{fig:eq1}. Once we have $k(v)$, we compute $\beta(v)$ and $\alpha(v)$ using
    \begin{align*}
        \beta(v) &= \int_0^v \frac{n-1}{2} \frac{s k(y)^{s-1} k'(y) + t y^{t-1}}{2} \\
        & \qquad \qquad \qquad \qquad \qquad \qquad \left(\frac{k(y)^s+y^t}{2}\right)^{n-2} y d y, \\
        \alpha(v) &= \begin{cases} \beta(k^{-1}(v)), \qquad \qquad \qquad \qquad \ \ \ \text{$0 \le v \le k(1)$} \\
                \beta(1) + \int_{k(1)}^{v} (n-1) \frac{s y^{s-1}}{2} \left(\frac{y^s+1}{2}\right)^{n-2} y dy, \\
                \qquad \qquad \qquad \qquad \qquad \qquad \quad \ \ \text{$k(1) \le v \le 1$}
                \end{cases}
    \end{align*}
    given in Figure~\ref{fig:eq2}. Further, we can use it to calculate the expected output of the players, for a player in the target group it is $\bE_{v \sim F}[\alpha(v)] \approx 0.103$ and for a player in the non-target group player it is $\bE_{v \sim G}[\beta(v)] \approx 0.052$.
    
\begin{figure}[t]

\includegraphics[width=\columnwidth]{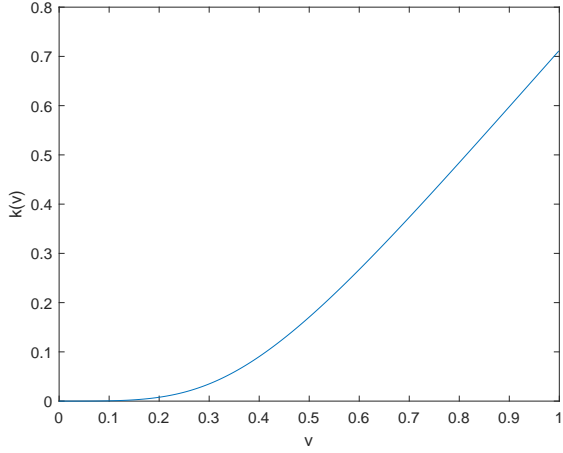} 
\caption{$k(v)$ vs $v$ (Example~\ref{ex:eq:1}).}
\label{fig:eq1}
\end{figure}

\begin{figure}[t]

\includegraphics[width=\columnwidth]{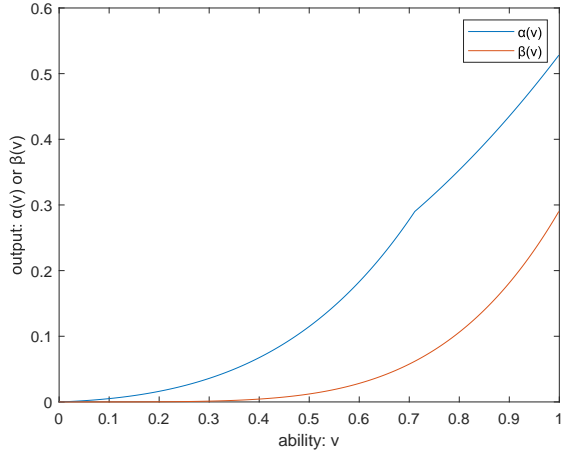} 
\caption{$\alpha(v)$ and $\beta(v)$ vs $v$ (Example~\ref{ex:eq:1}).}
\label{fig:eq2}
\end{figure}

\end{example}

\subsection{Further Details on Examples from Section~\ref{sec:desGen}}
The examples from Section~\ref{sec:desGen} have be repeated here with more details.

\begin{example}[Example~\ref{ex:des:0}]
Let $F(x) = 1 - (1-x)^{n-1}$, $G(x) = ((n-1)x - F(x))/(n-2)$, and $\mu = 1/(n-1)$. We get $H(x) = \mu F(x) + (1-\mu) G(x) = x$, which is the uniform distribution over $[0,1]$. Note that 
\[
    \psi(x) = \frac{x(1-F(x))}{1-H(x)} = \frac{x(1-x)^{n-1}}{(1-x)} = x(1-x)^{n-2}.
\]
It can be checked that the order-statistic that maximizes the inner-product with $\psi(x)$ is $f^H_{n,n-1}(x) = n(n-1)x(1-x)^{n-1}$. So, the optimal value is $k^* = n-2$. The total output generated for $k^* = n-2$ is
\begin{multline*}
    \frac{1}{n} \int_0^1 x f^{H}_{n,n-1}(x) dx \\
    = \frac{2}{n(n+1)} \int_0^1 f^{H}_{n+1,n-1}(x) dx = \frac{2}{n(n+1)}.
\end{multline*}
On the other hand, we know that the expected total output is maximized by $j=1$ and is equal to
\begin{multline*}
    \frac{1}{n} \int_0^1 x f^{H}_{n,2}(x) dx \\
    = \frac{(n-1)}{n(n+1)} \int_0^1 f^{H}_{n+1,2}(x) dx = \frac{(n-1)}{n(n+1)}.
\end{multline*}
The ratio between the two quantities is $(n-1)/2 = \Omega(n)$.
\end{example}


\begin{example}[Example~\ref{ex:des:2}]
Let $\mu = 1/8$, $F(x) = 1 - S(x)$, $G(x) = (x-\mu F(x))/(1-\mu)$, where $S(x)$ is given by:
\[
    S(x) = \begin{cases} 1, &\text{ if $x < 3/4$} \\
    16(1-x)^2, &\text{ if $3/4 \le x < 15/16$} \\
    1-x, &\text{ if $x \ge 15/16$}
    \end{cases}.
\]
It can be verified that $F(x)$ and $G(x)$ are continuous cumulative distribution functions. The distribution of the general population is $H(x) = \mu F(x) + (1-\mu) G(x) = x$. Also, $F(x) \le H(x)$ for every $x$ (and strict for some values of $x$), and therefore, $F$ FOS dominates $H$. The objective value is:
\begin{align*}
    &\frac{1}{j} \left( \int_0^{3/4} x f^H_{n-1,j}(x) dx \right. \\
    & \quad + \int_{3/4}^{15/16} 16x(1-x)^2 f^H_{n-1,j}(x) dx \\
    & \quad \left. + \int_{15/16}^1 x(1-x) f^H_{n-1,j}(x) dx \right) \\
    =& \ \frac{n-j}{j \cdot n} F^H_{n,j}\left(3/4\right)\\
    +& \frac{16(n-j)(j+1)}{n(n+1)(n+2)} (F^H_{n+2,j+2}(15/16) - F^H_{n+2,j+2}(3/4)) \\
    & \quad + \frac{n-j}{n(n+1)} \left(1-F^H_{n+1,j+1}\left(15/16\right)\right),
\end{align*}
where $F^H_{n,j}$ is the CDF corresponding to the PDF $f^H_{n,j}$.
Note that $F^H_{n,j}$ is equal to the Beta distribution with parameters $n+1-j$ and $j$. The optimal value for $j$ may be derived with an involved calculation, but for our purpose, we just want to illustrate that the optimal $j$ is not $1$; it can be checked that for $n=50$, the optimal value for $j$ is $11$ and the optimal objective value is approximately $0.0498$.
\end{example}

\begin{example}[Example~\ref{ex:des:3}]
Let $\mu = 2/3$, $F(x) = 3x^2 - 2x^3$, and $G(x) = 3x - 6x^2 + 4x^3$. Observe that $H(x) = \mu F(x) + (1-\mu) G(x) = x$, and the variance of $H$ is $1/12$ while the variance of $F$ is $1/20$, and both have mean $1/2$. The objective value is given by:
\begin{multline*}
    \frac{1}{j} \int_0^1 x(1-3x^2+2x^3) f^H_{n-1,j}(x) dx \\
     = \frac{n-j}{j\cdot n} \left( 1 - \frac{(n+1-j)(n+2-j)}{(n+1)(n+2)} \right. \\
    \left. \left(3 - 2 \cdot \frac{n+3-j}{n+3} \right) \right).
\end{multline*}
Solving for $j$, we get $j = \frac{5n}{6} - \frac{\sqrt{7n^2+30n+39}}{6}+\frac{1}{2}$; as $j$ is integer, either the floor or the ceiling of this number is optimal. For example, for $n=50$, the optimal value is $19$. 
\end{example}

\begin{example}[Example~\ref{ex:des:4}]
 Let $\mu = 1/4$, $F(x) = 1-S(x)$, and $G(x) = (x-\mu F(x))/(1-\mu)$, where $S(x)$ is given by:
\[
    S(x) = \begin{cases} 1-48x/31, &\text{ if $x < 31/96$} \\
    1/2, &\text{ if $31/96 \le x < 3/4$} \\
    8(1-x)^2, &\text{ if $3/4 \le x < 7/8$} \\
    1-x, &\text{ if $x \ge 7/8$}
    \end{cases}
\]
It can be checked that $F(x)$ and $G(x)$ are valid and for the distribution of the general population we have $H(x) = \mu F(x) + (1-\mu) G(x) = x$. It can also be checked that the mean of all the distributions is $1/2$, and the variance of $H$ is $1/12 \approx 0.083$ while the variance of $F$ is $6703/55296 \approx 0.121 > 1/12$. The objective value is:
\begin{align*}
    &\frac{1}{j} \left( \int_0^{31/96} x \left(1-\frac{48}{31}x\right) f^H_{n-1,j}(x) dx \right. \\
    &\qquad + \int_{31/96}^{3/4} x\cdot\frac{1}{2}\cdot f^H_{n-1,j}(x) dx \\
    &\qquad + \int_{3/4}^{7/8} x\cdot 8\cdot (1-x)^2 f^H_{n-1,j}(x) dx \\
    &\qquad \left. + \int_{7/8}^{1} x (1-x) f^H_{n-1,j}(x) dx \right) \\
    &= \frac{n-j}{2\cdot j\cdot n} \left(F^H_{n,j}\left(31/96\right) + F^H_{n,j}\left(3/4\right)\right) \\
    &\qquad - \frac{48(n-j)(n+1-j)}{31 \cdot j \cdot n \cdot (n+1)} F^H_{n+1,j}\left(31/96\right) \\
    &\quad+ \frac{8(n-j)(j+1)}{n(n+1)(n+2)} (F^H_{n+2,j+2}(7/8) - F^H_{n+2,j+2}(3/4)) \\
    & \qquad + \frac{n-j}{n(n+1)} \left(1-F^H_{n+1,j+1}\left(7/8\right)\right).
\end{align*}
Note that $F^H_{n,j}$ is the Beta distribution with parameters $n+1-j$ and $j$. 
For $n=50$, the optimal value for $j$ is $11$ and the optimal objective value is approximately $0.0249$.
\end{example}

\subsection{Comparing Group-Specific and General Prizes}\label{sec:desCompare}
In this section, we compare the choice among only group-specific prizes and only general prizes. This choice depends upon the distributions $F$ and $G$, and the parameter $\mu$ that captures the relative frequency of target-group agents. Even if we factor out the effect of the distributions $F$ and $G$, then also the choice among group-specific and general prizes depends upon $\mu$. If $\mu$ is high enough, then it is better to have group-specific prizes. On the other hand, if $\mu$ is quite low, then having only group-specific prizes leads to very low competition, and therefore, a low output is produced by the agents, and using general prizes is better than group-specific prizes. The example below illustrates this.

\begin{example}
Let $F(x) = G(x) = x$. As shown in Theorem~\ref{thm:desTarget}, it is optimal to allocate the entire budget to the first prize for group-specific prizes. For $F=G$, a similar result holds for general prizes as well.

We compare the following three cases:
\begin{enumerate}
    \item[A] Only general prizes of value $1$; $w_1 = 1$, $w_j = 0$ for $j > 1$, $\omega_j = 0$ for $j \in [n]$.
    \item[B] Only target group-specific prizes of value $1$; $w_j = 0$ for $j \in [n]$, $\omega_1 = 1$, $\omega_j = 0$ for $j > 1$.
    \item[C] Only target group-specific prizes of value $\mu$; $w_j = 0$ for $j \in [n]$, $\omega_1 = \mu$, $\omega_j = 0$ for $j > 1$.
\end{enumerate}
As $\mu$ is the fraction of target group agents in the population, the contest organizer may want to have a contest for the target group agents with a prize of $\mu$. This distributes the prize of $1$ equally between the target group and the non-target group. Case (C) above captures this scenario. 

Let $\alpha_A(v)$, $\alpha_B(v)$, and $\alpha_C(v)$ be the output generated by a player with ability $v$ for the three cases A, B, and C, respectively. 

The expected output for case (A) general prizes is:
\begin{multline*}
    \bE_{v \sim F}[\alpha_A(v)] = \int_0^1 v (1-F(v)) f^F_{n-1,1}(v) dv \\
    =  (n-1) \int_0^1 (1-v) v^{n-1} dv = \frac{n-1}{n(n+1)}.
\end{multline*}

The expected output for case (B) group-specific prizes of value $1$ is:
\begin{multline*}
    \bE_{v \sim F}[\alpha_B(v)] = \int_0^1 v (1-F(v)) f^{1-\mu + \mu F}_{n-1,1}(v) dv \\
    = (n-1) \int_0^1 v (1-v) (1 - \mu + \mu v)^{n-2} \mu dv.
\end{multline*}
Let $H(x) = 1 - \mu + \mu F(x) = 1-\mu + \mu x$. We have $(1-H(x)) = \mu (1-x)$ and $h(x) = \mu$. Let us write $f^H_{n,j}(x)$ as $h_{n,j}(x)$ and $F^H_{n,j}(x)$ as $H_{n,j}(x)$. Now,
\[
    v = \frac{1}{\mu} \mu v = \frac{1}{\mu} ( (1-\mu + \mu v) - (1-\mu)) = \frac{1}{\mu} (H(v) - (1-\mu)).
\]
Incorporating this into the expected output formula, we get
\begin{align*}
    &\bE_{v \sim F}[\alpha_B(v)] \\
    =& (n-1) \int_0^1 v (1-v) (1 - \mu + \mu v)^{n-2} \mu dv \\
    =& \frac{(n-1)}{\mu^2} \\
    &\qquad \int_0^1 (H(v) - (1-\mu)) (1-H(v)) H(v)^{n-2} h(v) dv  \\
    =& \frac{1}{\mu^2} \int_0^1 (H(v) - (1-\mu)) (1-H(v)) h_{n-1,1}(v)  dv \\
    =& \frac{1}{n \mu^2} \int_0^1 (H(v) - (1-\mu)) h_{n,2}(v)  dv  \\
    =& \frac{1}{n \mu^2} \left( \int_0^1 H(v) h_{n,2}(v) dv - (1-\mu) \int_0^1 h_{n,2}(v) dv \right) \\
    =& \frac{1}{n \mu^2} \bigg( \frac{n-1}{n+1} \int_0^1 h_{n+1,2}(v) dv - (1-\mu) (1- H_{n,2}(0)) \bigg) \\
    =& \frac{1}{n \mu^2} \bigg( \frac{n-1}{n+1} (1-H_{n+1,2}(0)) - (1-\mu) (1- H_{n,2}(0)) \bigg).
\end{align*}
Now, $H_{n,2}(0) = H(0)^n + n (1-H(0)) H(0)^{n-1} = (1-\mu)^n + n \mu (1-\mu)^{n-1} = (1-\mu)^{n-1}(1-\mu + n \mu)$. Similarly, $H_{n+1,2}(0) = (1-\mu)^n(1-\mu + (n+1) \mu)$. Replacing this, we get
\begin{align*}
&= \frac{1}{n \mu^2} \bigg( \frac{n-1}{n+1} (1-(1-\mu)^n(1-\mu + (n+1) \mu))  \\
    &\qquad \qquad \qquad - (1-\mu) (1 - (1-\mu)^{n-1}(1-\mu + n \mu)) \bigg) \\
&= \frac{1}{n (n+1) \mu^2} ( (n-1) - (1-\mu)(n+1) \\
    & \qquad \qquad \qquad + (1-\mu)^n (2(1-\mu) + (n+1)\mu) ).
\end{align*}
Although the solution above may seem difficult to interpret, but one can observe that as $\mu \rightarrow 0$, it goes to $0$, but as $\mu \rightarrow 1$, it reaches $(n-1)/(n(n+1))$ from above.

\begin{figure}[t]

\includegraphics[width=\columnwidth]{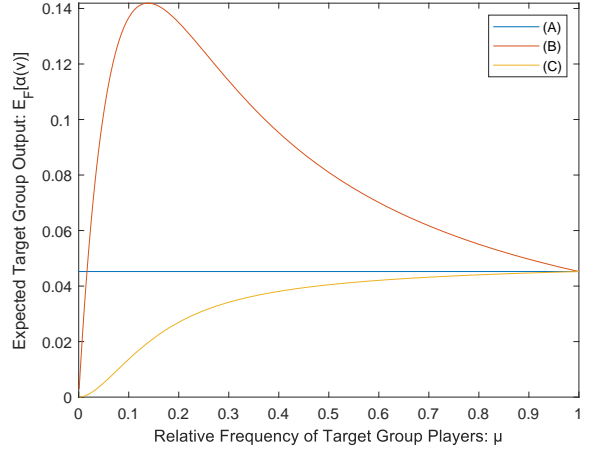} 
\caption{Expected Target Group Output vs Relative Frequency of Target Group Players. (A) Only general prizes of value $1$. (B) Only target group-specific prizes of value $1$. (C) Only target group-specific prizes of value $\mu$. }
\label{fig1}
\end{figure}

We can observe that the expected output for case (C) group-specific prizes of value $\mu$ is equal to $\mu$ times the expected output for case (B):
\begin{multline*}
    \bE_{v \sim F}[\alpha_C(v)] = \mu \bE_{v \sim F}[\alpha_B(v)] \\
    = \frac{1}{n (n+1) \mu} ( (n-1) - (1-\mu)(n+1) \\
         + (1-\mu)^n (2(1-\mu) + (n+1)\mu) ).
\end{multline*}
The derivative of the above expression is non-negative in $\mu$ and so it is increasing, for $\mu \in [0,1]$. For $\mu \rightarrow 0$, it goes to $0$, and for $\mu \rightarrow 1$, it goes to $(n-1)/(n(n+1))$. So, $\bE_{v \sim F}[\alpha_C(v)] \le (n-1)/(n(n+1))$ for $\mu \in [0,1]$.

Figure~\ref{fig1} plots the expected output for the three cases for $n=20$. 

We observe that case (A) always dominates case (C), so a general prize of $1$ is better than a group-specific prize of $\mu$, although in both cases the amount of prize in the contest per target group agent is the same. This happens because of higher competition in case (A): for general prizes, the target group agents produce higher output to compete with the non-target group agents.

Comparing case (A) general prize of $1$ and (B) group-specific prize of $1$, we see that if $\mu$ is very small, then case (A) dominates case (B), and the opposite is true if $\mu$ is sufficiently large. This happens because of too little competition in case (B) for very small $\mu$.


\end{example}

\end{document}